\documentclass[letterpaper,11pt]{article}

\usepackage{amsmath}
\usepackage{amsfonts}
\usepackage{amssymb}
\usepackage{amsthm}
\usepackage{url,ifthen}
\usepackage{srcltx}
\usepackage{multirow}
\usepackage{boxedminipage}
\usepackage[margin=1.15in]{geometry}
\usepackage{nicefrac}
\usepackage{xspace}
\usepackage{graphicx}
\usepackage{color}
\definecolor{DarkGreen}{rgb}{0.1,0.5,0.1}
\definecolor{DarkRed}{rgb}{0.5,0.1,0.1}
\definecolor{DarkBlue}{rgb}{0.1,0.1,0.5}

\usepackage[small]{caption}

\usepackage[pdftex]{hyperref}
\hypersetup{
    unicode=false,          
    pdftoolbar=true,        
    pdfmenubar=true,        
    pdffitwindow=false,      
    pdfnewwindow=true,      
    colorlinks=true,       
    linkcolor=DarkBlue,          
    citecolor=DarkGreen,        
    filecolor=DarkRed,      
    urlcolor=DarkBlue,          
    pdftitle={},
    pdfauthor={},    
}

\def\draft{1} 

\def\submit{0} 

\ifnum\draft=1 
    \def\ShowAuthNotes{1}
\else
    \def\ShowAuthNotes{0}
\fi

\ifnum\submit=1
\newcommand{\forsubmit}[1]{#1}
\newcommand{\forreals}[1]{}
\else
\newcommand{\forreals}[1]{#1}
\newcommand{\forsubmit}[1]{}
\fi

\ifnum\ShowAuthNotes=1
\newcommand{\authnote}[2]{{ \footnotesize \bf{\color{DarkRed}[#1's Note:
{\color{DarkBlue}#2}]}}}
\else
\newcommand{\authnote}[2]{}
\fi

\newtheorem{theorem}{Theorem}[section]

\newtheorem{lemma}[theorem]{Lemma}
\newtheorem{corollary}[theorem]{Corollary}

\newtheorem{claim}[theorem]{Claim}

\newtheorem{conjecture}[theorem]{Conjecture}
\theoremstyle{definition}
\newtheorem{definition}[theorem]{Definition}

\newcommand{\chapterref}[1]{\hyperref[ch:#1]{Chapter~\ref{ch:#1}}}

\newcommand{\claimref}[1]{\hyperref[claim:#1]{Claim~\ref{claim:#1}}}
\newcommand{\corollarylabel}[1]{\label{cor:#1}}
\newcommand{\corollaryref}[1]{\hyperref[cor:#1]{Corollary~\ref{cor:#1}}}
\newcommand{\definitionlabel}[1]{\label{def:#1}}
\newcommand{\definitionref}[1]{\hyperref[def:#1]{Definition~\ref{def:#1}}}
\newcommand{\equationlabel}[1]{\label{eq:#1}}
\newcommand{\equationref}[1]{\hyperref[eq:#1]{Equation~\ref{eq:#1}}}

\newcommand{\factref}[1]{\hyperref[fact:#1]{Fact~\ref{fact:#1}}}
\newcommand{\figurelabel}[1]{\label{fig:#1}}
\newcommand{\figureref}[1]{\hyperref[fig:#1]{Figure~\ref{fig:#1}}}
\newcommand{\itemlabel}[1]{\label{item:#1}}
\newcommand{\itemref}[1]{\hyperref[item:#1]{Item~(\ref{item:#1})}}
\newcommand{\lemmalabel}[1]{\label{lem:#1}}
\newcommand{\lemmaref}[1]{\hyperref[lem:#1]{Lemma~\ref{lem:#1}}}

\newcommand{\propref}[1]{\hyperref[prop:#1]{Proposition~\ref{prop:#1}}}

\newcommand{\propositionref}[1]{\hyperref[prop:#1]{Proposition~\ref{prop:#1}}}

\newcommand{\remarkref}[1]{\hyperref[rem:#1]{Remark~\ref{rem:#1}}}
\newcommand{\sectionlabel}[1]{\label{sec:#1}}
\newcommand{\sectionref}[1]{\hyperref[sec:#1]{Section~\ref{sec:#1}}}
\newcommand{\theoremlabel}[1]{\label{thm:#1}}
\newcommand{\theoremref}[1]{\hyperref[thm:#1]{Theorem~\ref{thm:#1}}}

\usepackage[T1]{fontenc}
\usepackage{kpfonts}
\usepackage{microtype}

\newcommand{\Esymb}{\mathbb{E}}
\newcommand{\Psymb}{\mathbb{P}}
\newcommand{\Vsymb}{\mathbb{V}}
\DeclareMathOperator*{\E}{\Esymb}
\DeclareMathOperator*{\Var}{\Vsymb}
\DeclareMathOperator*{\ProbOp}{\Psymb r}
\renewcommand{\Pr}{\ProbOp}

\newcommand{\mper}{\,.}
\newcommand{\mcom}{\,,}

\renewcommand{\hat}{\widehat}

\newcommand{\cD}{{\cal D}}

\newcommand{\cP}{{\cal P}}

\newcommand{\cX}{{\cal X}}

\newcommand{\defeq}{\stackrel{\small \mathrm{def}}{=}}
\renewcommand{\leq}{\leqslant}
\renewcommand{\le}{\leqslant}
\renewcommand{\geq}{\geqslant}
\renewcommand{\ge}{\geqslant}

\newcommand{\Set}[1]{\left\{#1\right\}}

\newcommand{\norm}[1]{\lVert#1\rVert}
\newcommand{\Norm}[1]{\left\lVert#1\right\rVert}

\newcommand{\signs}{\{-1,1\}}

\newcommand{\R}{\mathbb{R}}

\usepackage{bm}

\newcommand{\sign}{\mathit{sign}}

\newcommand{\supp}{\mathrm{supp}}

\renewcommand{\epsilon}{\varepsilon}

\newcommand{\eps}{\epsilon}

\newcommand{\remove}[1]{}

\DeclareMathOperator*{\range}{range}

\newcommand{\ignore}[1]{}

\newcommand{\tr}{\mathrm{tr}}
\newcommand{\trans}{\top}

\newcommand{\PPM}{{\sc PPM}\xspace}
\newcommand{\NPM}{{\sc NPM}\xspace}
\newcommand{\SPM}{{\sc SPM}\xspace}

\usepackage{enumitem}
\newenvironment{itm}
{\begin{itemize}[noitemsep,topsep=0pt,parsep=0pt,partopsep=0pt]}
{\end{itemize}}

\newenvironment{enum}
{
\begin{enumerate}[noitemsep,topsep=0pt,parsep=0pt,partopsep=0pt]}
{\end{enumerate}}

\title{The Noisy Power Method:\\A Meta Algorithm with Applications}

\author{Moritz Hardt\thanks{IBM Research Almaden. Email: {\tt
mhardt@us.ibm.com}}
 \and Eric Price\thanks{IBM Research Almaden. Email: {\tt
ecprice@mit.edu}}
}

\begin{document}
\maketitle
\begin{abstract}
  We provide a new robust convergence analysis of the well-known power
  method for computing the dominant singular vectors of a matrix that
  we call the \emph{noisy power method}. Our result characterizes the
  convergence behavior of the algorithm when a significant amount
  noise is introduced after each matrix-vector multiplication. The
  noisy power method can be seen as a meta-algorithm that has recently
  found a number of important applications in a broad range of machine
  learning problems including alternating minimization for matrix
  completion, streaming principal component analysis (PCA), and
  privacy-preserving spectral analysis.
  Our general analysis subsumes several existing ad-hoc convergence
  bounds and resolves a number of open problems in multiple
  applications:

  {\bf Streaming PCA.} A recent work of Mitliagkas et al.~(NIPS 2013)
  gives a space-efficient algorithm for PCA
  in a streaming model where samples are drawn from a gaussian spiked
  covariance model. We give a simpler and more general analysis that
  applies to arbitrary distributions confirming experimental evidence
  of Mitliagkas et al. Moreover, even in the spiked covariance model
  our result gives quantitative improvements in a natural parameter
  regime. It is also notably simpler and follows easily from our
  general convergence analysis of the noisy power method together with
  a matrix Chernoff bound.

  {\bf Private PCA.}  We provide the first nearly-linear time
  algorithm for the problem of differentially private principal
  component analysis that achieves nearly tight worst-case error
  bounds. Complementing our worst-case bounds, we show that the error
  dependence of our algorithm on the matrix dimension can be replaced
  by an essentially tight dependence on the \emph{coherence} of the
  matrix. This result resolves the main problem left open by Hardt and
  Roth (STOC 2013). The coherence is always bounded by the matrix
  dimension but often substantially smaller thus leading to
  strong average-case improvements over the optimal worst-case bound.
\end{abstract}

%
%
%

\section{Introduction}

Computing the dominant singular vectors of a matrix is one of the most
important algorithmic tasks underlying many applications including
low-rank approximation, PCA, spectral clustering, dimensionality
reduction, matrix completion and topic modeling. The classical problem
is well-understood, but many recent applications in machine learning
face the fundamental problem of approximately finding singular vectors
in the presence of noise. Noise can enter the computation through a
variety of sources including sampling error, missing entries,
adversarial corruptions and privacy constraints. It is desirable to
have one robust method for handling a variety of cases without the
need for ad-hoc analyses. In this paper we consider the \emph{noisy
  power method}, a fast general purpose method for computing the
dominant singular vectors of a matrix when the target matrix can only
be accessed through inaccurate matrix-vector products.

\figureref{NPM} describes the method when the target matrix~$A$ is a
symmetric $d\times d$ matrix---a generalization to asymmetric matrices
is straightforward. The algorithm starts from an initial matrix
$X_0\in\R^{d\times p}$ and iteratively attempts to perform the update
rule $X_\ell \to AX_\ell.$ However, each such matrix product is
followed by a possibly adversarially and adaptively chosen
perturbation~$G_\ell$ leading to the update rule $X_\ell \to AX_\ell +
G_\ell.$ It will be convenient though not necessary to maintain that
$X_\ell$ has orthonormal columns which can be achieved through a
QR-factorization after each update.

\begin{figure}[h]
\centering
\begin{boxedminipage}{0.8\textwidth}
\noindent \textbf{Input:} Symmetric matrix $A\in\mathbb{R}^{d\times d},$ number of
iterations $L,$ dimension $p$
\begin{enum}
\item 
Choose $X_0\in\R^{d\times p}.$
\item For $\ell = 1$ to $L$:
\begin{enum}
\item\itemlabel{mult} $Y_\ell \leftarrow AX_{\ell-1}+G_\ell$
where $G_\ell\in\R^{d\times p}$ is some perturbation
\item Let $Y_\ell = X_\ell R_\ell$ be a QR-factorization of $Y_\ell$
\end{enum}
\end{enum}
\noindent \textbf{Output:} Matrix $X_L$ 
\end{boxedminipage}
\vspace{-2mm}
\caption{Noisy Power Method (\NPM)}
\figurelabel{NPM}
\end{figure}
The noisy power method is a meta algorithm that when instantiated with
different settings of $G_\ell$ and $X_0$ adapts to a variety of
applications.  In fact, there have been a number of recent surprising
applications of the noisy power method:
\begin{enumerate}
\item Jain et al.~\cite{JainNS13,Hardt14} observe that the update rule of the
  well-known alternating least squares heuristic for matrix completion
  can be considered as an instance of \NPM. This lead to the first
  provable convergence bounds for this important heuristic.
\item Mitgliakas et al.~\cite{MCJ13} observe that \NPM
  applies to a streaming model of principal component analysis (PCA)
  where it leads to a space-efficient and practical algorithm for PCA
  in settings where the covariance matrix is too large to process
  directly.
\item Hardt and Roth~\cite{HardtR13} consider the power method in the
  context of privacy-preserving PCA where noise is added to achieve
  differential privacy.
\end{enumerate}
In each setting there has so far only been an ad-hoc analysis of the
noisy power method. In the first setting, only local convergence is
argued, that is, $X_0$ has to be cleverly chosen. In the second
setting, the analysis only holds for the spiked covariance model of
PCA. In the third application, only the case $p=1$ was considered.

In this work we give a completely general analysis of the noisy power
method that overcomes limitations of previous analyses. Our result
characterizes the global convergence properties of the algorithm in
terms of the noise $G_\ell$ and the initial subspace $X_0$.  We then
consider the important case where $X_0$ is a randomly chosen
orthonormal basis.  This case is rather delicate since the initial
correlation between a random matrix $X_0$ and the target subspace is
vanishing in the dimension~$d$ for small $p.$ Another important
feature of the analysis is that it shows how $X_\ell$ converges
towards the first $k\le p$ singular vectors. Choosing~$p$ to be larger
than the target dimension leads to a quantitatively stronger
result. \theoremref{convergence} formally states our convergence
bound. Here we highlight one useful corollary to illustrate our more
general result.

\begin{corollary}
  \corollarylabel{random} \theoremlabel{randomconvergence} Let $k \leq
  p$.  Let $U\in\R^{d\times k}$ represent the top $k$ singular vectors
  of~$A$ and let $\sigma_1\ge\cdots\ge\sigma_n\ge0$ denote its
  singular values.  Suppose $X_0$ is an orthonormal basis of a random
  $p$-dimensional subspace.  Further suppose that at every step of
  \NPM we have
\[\textstyle
5\norm{G_\ell} \leq \eps (\sigma_k - \sigma_{k+1})
\quad\text{and}\quad 5\norm{U^\trans G_\ell} \leq (\sigma_k -
\sigma_{k+1}) \frac{\sqrt{p} - \sqrt{k-1}}{\tau \sqrt{d}}
\]
for some fixed parameter $\tau$ and $\eps < 1/2$.  Then with all but
$\tau^{-\Omega(p+1-k)} + e^{-\Omega(d)}$ probability, there exists an
$L = O( \frac{\sigma_k}{\sigma_k - \sigma_{k+1}}\log(d\tau/\epsilon))$ so that
after $L$ steps we have that $\Norm{(I-X_LX_L^\trans)U}\leq \epsilon.$
\end{corollary}

The corollary shows that the algorithm converges in the strong sense
that the entire spectral norm of $U$ up to an $\epsilon$ error is
contained in the space spanned by~$X_L.$ To achieve this the result
places two assumptions on the magnitude of the noise. The total
spectral norm of $G_\ell$ must be bounded by $\epsilon$ times the
separation between $\sigma_k$ and $\sigma_{k+1}.$ This dependence on
the singular value separation arises even in the classical
perturbation theory of Davis-Kahan~\cite{DavisK70}. The second condition
is specific to the power method and requires that the noise term is
proportionally smaller when projected onto the space spanned by the
top $k$ singular vectors. This condition ensures that the correlation
between $X_\ell$ and $U$ that is initially very small is not destroyed
by the noise addition step.  If the noise term has some spherical
properties (e.g. a Gaussian matrix), we expect the projection onto $U$
to be smaller by a factor of $\sqrt{k/d},$ since the space $U$ is
$k$-dimensional. In the case where $p=k+\Omega(k)$ this is precisely
what the condition requires. When $p=k$ the requirement is stronger by
a factor of~$k.$ This phenomenon stems from the fact that the smallest
singular value of a random $p\times k$ gaussian matrix behaves
differently in the square and the rectangular case.

We demonstrate the usefulness of our convergence bound with several
novel results in some of the aforementioned applications.

\subsection{Application to memory-efficient streaming PCA}

In the streaming PCA setting we receive a stream of samples
$z_1,z_2,\dots z_n\in\R^d$ drawn i.i.d.~from an unknown
distribution~$\cD$ over $\R^d.$ Our goal is to compute the dominant
$k$ eigenvectors of the covariance matrix $A=\E_{z\sim\cD} zz^\trans.$
The challenge is to do this in space linear in the output size, namely
$O(kd).$
Recently, Mitgliakas et al.~\cite{MCJ13} gave an algorithm for this
problem based on the noisy power method. We analyze the same
algorithm, which we restate here and call \SPM:
\begin{figure}[h]
\centering
\begin{boxedminipage}{0.8\textwidth}
  \noindent \textbf{Input:} Stream of samples
  $z_1,z_2,\dots,z_n\in\R^d,$ iterations $L,$ dimension~$p$
\begin{enum}
\item Let $X_0\in\R^{d\times p}$ be a random orthonormal basis.  Let $T=
  \lfloor m/L\rfloor$
\item For $\ell = 1$ to $L$:
\begin{enum}
\item Compute $Y_\ell = A_\ell X_{\ell-1}$ where
  $A_\ell=\sum_{i=(\ell-1)T+1}^{\ell T}z_iz_i^\trans$
\item Let $Y_\ell = X_\ell R_\ell$ be a QR-factorization of $Y_\ell$
\end{enum}
\end{enum}
\noindent \textbf{Output:} Matrix $X_L$ 
\end{boxedminipage}
\caption{Streaming Power Method (\SPM)}
\figurelabel{SPM}
\end{figure}

The algorithm can be executed in space $O(pd)$ since the update step
can compute the $d\times p$ matrix $A_\ell X_{\ell-1}$ incrementally
without explicitly computing $A_\ell.$ The algorithm maps to our
setting by defining $G_\ell = (A_\ell-A)X_{\ell-1}.$ With this
notation $Y_\ell = AX_{\ell-1} + G_\ell.$ We can apply
\corollaryref{random} directly once we have suitable bounds on
$\|G_\ell\|$ and $\|U^\trans G_\ell\|.$

The result of~\cite{MCJ13} is specific to the spiked covariance model.
The spiked covariance model is defined by an orthonormal basis
$U\in\R^{d\times k}$ and a diagonal matrix $\Lambda\in\R^{k\times k}$
with diagonal entries $\lambda_1\ge\lambda_2\ge\cdots\ge\lambda_k>0.$
The distribution $\cD(U,\Lambda)$ is defined as the normal
distribution $\mathrm{N}(0,(U\Lambda^2 U^\trans +
\sigma^2\mathrm{Id}_{d\times d})).$ Without loss of generality we can
scale the examples such that $\lambda_1=1.$ 
One corollary of our result shows that the algorithm outputs $X_L$
such that $\Norm{(I-X_LX_L^\trans)U}\le\epsilon$ with probability
$9/10$ provided $p=k+\Omega(k)$ and the number of samples satisfies
\[
n = \Theta\left(\frac{\sigma^6+1}{\epsilon^2 \lambda_k^6}kd\right).
\]
Previously, the same bound\footnote{That the bound stated
  in~\cite{MCJ13} has a $\sigma^6$ dependence is not completely
  obvious. There is a $O(\sigma^4)$ in the numerator and
  $\log((\sigma^2+0.75\lambda_k^2)/(\sigma^2+0.5\lambda_k^2))$ in the
  denominator which simplifies to $O(1/\sigma^2)$ for constant
  $\lambda_k$ and $\sigma^2\geq 1.$} was known with a quadratic
dependence on $k$ in the case where $p=k.$ Here we can strengthen the
bound by increasing $p$ slightly.

While we can get some improvements even in the spiked covariance
model, our result is substantially more general and applies to any
distribution. The sample complexity bound we get varies according to a
technical parameter of the distribution.  Roughly speaking, we get a
near linear sample complexity if the distribution is either ``round''
(as in the spiked covariance setting) or is very well approximated by
a $k$ dimensional subspace.  To illustrate the latter condition, we
have the following result without making any assumptions other than
scaling the distribution:
\begin{corollary}
  Let $\cD$ be any distribution scaled so that
  $\Pr\Set{\|z\|>t}\le\exp(-t)$ for every $t\ge 1.$ Let $U$ represent
  the top $k$ eigenvectors of the covariance matrix $\E zz^\trans$ and
  $\sigma_1\ge\cdots\ge\sigma_d\ge0$ its eigenvalues. Then, \SPM
  invoked with $p=k+\Omega(k)$ outputs a matrix $X_L$ such with
  probability $9/10$ we have $\Norm{(I-X_LX_L^\trans)U}\le\epsilon$
  provided \SPM receives $n$ samples where $n$ satisfies $n = \tilde
  O\left(\frac{\sigma_k}{\epsilon^2k(\sigma_k-\sigma_{k+1})^3}\cdot
    d\right)\mper$
\end{corollary}
The corollary establishes a sample complexity that's linear in~$d$
provided that the spectrum decays quickly, as is common in
applications. For example, if the spectrum follows a power law so that
$\sigma_j\approx j^{-c}$ for a constant $c>1/2,$ the bound becomes
$n=\tilde O(k^{2c+2}d/\epsilon^2).$

\subsection{Application to privacy-preserving spectral analysis}
Many applications of singular vector computation are plagued by the
fact that the underlying matrix contains sensitive information about
individuals. A successful paradigm in privacy-preserving data analysis rests
on the notion of \emph{differential privacy} which requires all access to the
data set to be randomized in such a way that the presence or absence of a
single data item is hidden. The notion of data item
varies and could either refer to a single entry, a single row, or a
rank-$1$ matrix of bounded norm. More formally, Differential Privacy
requires that the output distribution of the algorithm changes only
slightly with the addition or deletion of a single data item.  This
requirement often necessitates the introduction of significant levels
of noise that make the computation of various objectives
challenging. Differentially private singular vector computation has
been studied actively in recent years~\cite{BlumDMN05,McSherryM09,BlockiBDS12,ChaudhuriSS12,KapralovT13,HardtR12,HardtR13,DworkTTZ14}. 
There are two main objectives. The first is computational
efficiency. The second objective is to minimize the amount
of error that the algorithm introduces.

In this work, we give a fast algorithm for differentially private
singular vector computation based on the noisy power method that leads
to nearly optimal bounds in a number of settings that were considered in
previous work. The algorithm is described
in \figureref{PPM}. It's a simple instance of \NPM in which each noise
matrix $G_\ell$ is a gaussian random matrix scaled so that the
algorithm achieves $(\epsilon,\delta)$-differential privacy (as
formally defined in \definitionref{dp}). It is easy to see that the algorithm
can be implemented in time nearly linear in the number of nonzero entries of
the input matrix (input sparsity). This will later lead to strong improvements
in running time compared with several previous works.
\begin{figure}[h]
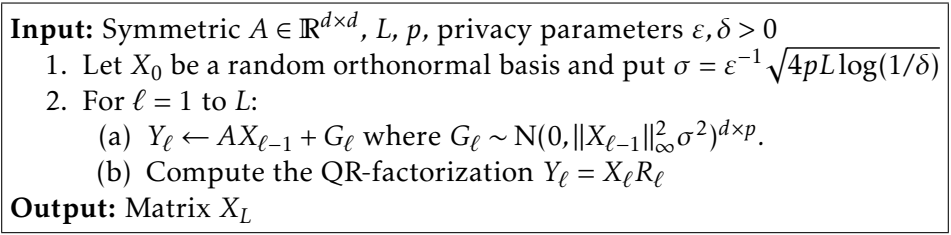

\centering
\begin{boxedminipage}{0.8\textwidth}
  \noindent \textbf{Input:} Symmetric $A\in\R^{d\times d},$ $L,$ $p,$ privacy
  parameters $\epsilon,\delta>0$
\begin{enum}
\item Let $X_0$ be a random orthonormal basis and put
$\sigma=\epsilon^{-1}\sqrt{4pL\log(1/\delta)}$
\item For $\ell = 1$ to $L$:
\begin{enum}
\item\itemlabel{mult} $Y_\ell \leftarrow AX_{\ell-1}+G_\ell$
where $G_\ell \sim \mathrm{N}(0,\|X_{\ell-1}\|_\infty^2\sigma^2)^{d\times p}.$
  \item Compute the QR-factorization $Y_\ell = X_\ell R_\ell$
\end{enum}
\end{enum}
\noindent \textbf{Output:} Matrix $X_L$ 
\end{boxedminipage}
\caption{Private Power Method (\PPM). Here $\|X\|_\infty=\max_{ij}|X_{ij}|.$}
\figurelabel{PPM}
\end{figure}

We first state a general purpose analysis of \PPM that follows from
\corollaryref{random}.
\begin{theorem}
  \theoremlabel{privacy-main} Let $k \leq p$.  Let $U\in\R^{d\times
    k}$ represent the top $k$ singular vectors of~$A$ and let
  $\sigma_1\ge\cdots\ge\sigma_d\ge0$ denote its singular values. Then,
  \PPM satisfies $(\epsilon,\delta)$-differential privacy and after $L
  = O( \frac{\sigma_k}{\sigma_k - \sigma_{k+1}}\log(d))$ iterations we have with
  probability $9/10$ that
\[
\Norm{(I-X_LX_L^\trans)U}\leq
O\left( 
\frac{\sigma\max\|X_\ell\|_\infty\sqrt{d\log L}}{\sigma_k-\sigma_{k+1}}
\cdot\frac{\sqrt{p}}{\sqrt{p}-\sqrt{k-1}}
\right)
\mper
\]
\end{theorem}
When $p=k+\Omega(k)$ the trailing factor becomes a constant. If $p=k$
it creates a factor $k$ overhead. In the worst-case we can always
bound $\|X_\ell\|_\infty$ by~$1$ since $X_\ell$ is an orthonormal
basis.  However, in principle we could hope that a much better bound
holds provided that the target subspace $U$ has small coordinates.
Hardt and Roth~\cite{HardtR12,HardtR13} suggested a way to accomplish
a stronger bound by considering a notion of \emph{coherence} of $A,$
denoted as $\mu(A).$ Informally, the coherence is a well-studied
parameter that varies between $1$ and $n,$ but is often observed to be
small. Intuitively, the coherence measures the correlation between the
singular vectors of the matrix with the standard basis. Low coherence
means that the singular vectors have small coordinates in the standard
basis. Many results on matrix completion and robust PCA crucially rely
on the assumption that the underlying matrix has low
coherence~\cite{CandesR09,CandesT10,CandesLMW11} (though the notion of
coherence here will be somewhat different).
\begin{theorem}\theoremlabel{privacy-mu}
Under the assumptions of \theoremref{privacy-main}, we have the conclusion
\[
\Norm{(I-X_LX_L^\trans)U}\leq
O\left( \frac{\sigma\sqrt{\mu(A)\log d\log L}}{\sigma_k-\sigma_{k+1}}
\cdot\frac{\sqrt{p}}{\sqrt{p}-\sqrt{k-1}}
\right)
\mper
\]
\end{theorem}
Hardt and Roth proved this result for the case where $p=1.$ The
extension to $p>1$ lost a factor of $\sqrt{d}$ in general and
therefore gave no improvement over \theoremref{privacy-main}. Our
result resolves the main problem left open in their work. The strength
of \theoremref{privacy-mu} is that the bound is essentially
dimension-free under a natural assumption on the matrix and never
worse than our worst-case result.  It is also known that in general
the dependence on $d$ achieved in \theoremref{privacy-main} is best
possible in the worst case (see discussion in~\cite{HardtR13}) so that
further progress requires making stronger assumptions. Coherence is a
natural such assumption. The proof of \theoremref{privacy-mu} proceeds
by showing that each iterate $X_\ell$ satisfies $\|X_\ell\|_\infty \le
O(\sqrt{\mu(A)\log(d)/d})$ and applying \theoremref{privacy-main}. To
do this we exploit a non-trivial symmetry of the algorithm that we
discuss in \sectionref{privacy-mu}.

\paragraph{Other objective functions and variants differential privacy.}
An important recent work by Dwork, Talwar, Thakurta and Zhang analyzes the
mechanism of adding Gaussian noise to the covariance matrix and computing a
truncated singular value decomposition of the noisy covariance
matrix~\cite{DworkTTZ14}.  Their objective function is a natural measure of
how much variance of the data is captured by the resulting subspace. Our
results are formally incomparable due to a different choice of objective
function. We also do not know how to analyze the performance of the power
method under their objective function.  Indeed, this is an interesting
question related to the content of Conjecture~\ref{conjecture:spectral} that
we will state shortly.

Our discussion above applied to $(\epsilon,\delta)$-differential
privacy under changing a single entry of the matrix. Several works
consider other variants of differential privacy. It is generally easy
to adapt the power method to these settings by changing the noise
distribution or its scaling. To illustrate this aspect, we consider
the problem of privacy-preserving principal component analysis as
recently studied by~\cite{ChaudhuriSS12,KapralovT13}. Both works
consider an algorithm called \emph{exponential mechanism}. The first
work gives a heuristic implementation that may not converge, while the
second work gives a provably polynomial time algorithm though the
running time is more than cubic. Our algorithm gives strong
improvements in running time while giving nearly optimal accuracy
guarantees as it matches a lower bound of \cite{KapralovT13} up to a $\tilde
O(\sqrt{k})$ factor. We also improve the error dependence on~$k$ by
polynomial factors compared to previous work. Moreover, we get an 
accuracy improvement of $O(\sqrt{d})$ for the case of $(\epsilon,\delta)$-differential
privacy, while these previous works only apply to $(\epsilon,0)$-differential
privacy.  \sectionref{PCA} provides formal statements.

\subsection{Related Work}

\paragraph{Numerical Analysis.} 
One might expect that a suitable analysis of the noisy power method would have
appeared in the numerical analysis literature. However, we are not aware of a
reference and there are a number of points to consider. First, our noise model
is adaptive thus setting it apart from the classical perturbation theory of
the singular vector decomposition~\cite{DavisK70}. Second, we think of the
perturbation at each step as large making it conceptually different from
floating point errors.  Third, research in numerical analysis over the past
decades has largely focused on faster Krylov subspace methods. There is some
theory of \emph{inexact Krylov methods}, e.g., \cite{SimonciniS07} that
captures the effect of noisy matrix-vector products in this context.  Related
to our work are also results on the perturbation stability of the
QR-factorization since those could be used to obtain convergence bounds for
subspace iteration. Such bounds, however, must depend on the condition number
of the matrix that the QR-factorization is applied to. See Chapter~19.9
in~\cite{Higham} and the references therein for background. Our proof strategy
avoids this particular dependence on the condition number. 

\paragraph{Streaming PCA.}
PCA in the streaming model is related to a host of well-studied
problems that we cannot survey completely here. We refer
to~\cite{ACLS12,MCJ13} for a thorough discussion of prior work. Not
mentioned therein is a recent work on incremental PCA~\cite{BDF13}
that leads to space efficient algorithms computing the top singular
vector; however, it's not clear how to extend their results to
computing multiple singular vectors.

\paragraph{Privacy.}
There has been much work on differentially private spectral analysis starting
with Blum et al.~\cite{BlumDMN05} who used an algorithm sometimes called
\emph{Randomized Response}, which adds a single noise matrix $N$ either to the
input matrix~$A$ or the covariance matrix $AA^\trans.$ This approach was used
by McSherry and Mironov~\cite{McSherryM09} for the purpose of a differentially
private recommender system. Most recently, as discussed earlier, Dwork,
Talwar, Thakurta and Zhang~\cite{DworkTTZ14} revisit (a variant of) the this
algorithm and give matching upper and lower bounds under a natural objective function. 
While often suitable when $AA^\trans$ fits into memory,
the approach can be difficult to apply when the dimension of $AA^\trans$
is huge as it requires computing a dense noise matrix~$N.$ The power method
can be applied more easily to large sparse matrices, as well as in a streaming
setting as shown by~\cite{MCJ13}.

Chaudhuri et al.~\cite{ChaudhuriSS12} and Kapralov-Talwar~\cite{KapralovT13}
use the so-called \emph{exponential mechanism} to sample approximate
eigenvectors of the matrix. The sampling is done using a heuristic approach
without convergence polynomial time convergence guarantees in the first case
and using a polynomial time algorithm in the second. Both papers achieve a
tight dependence on the matrix dimension~$d$ (though the dependence on~$k$ is
suboptimal in general).  Most closely related to our work are the results of
Hardt and Roth~\cite{HardtR13,HardtR12} that introduced matrix coherence as a
way to circumvent existing worst-case lower bounds on the error. They also
analyzed a natural noisy variant of power iteration for the case of computing
the dominant eigenvector of $A.$ When multiple eigenvectors are needed, their
algorithm uses the well-known deflation technique. However, this step loses
control of the coherence of the original matrix and hence results in
suboptimal bounds. In fact, a $\sqrt{\mathrm{rank}(A)}$ factor is lost. 

\subsection{Open Questions}

We believe \corollaryref{random} to be a fairly precise
characterization of the convergence of the noisy power method to the
top $k$ singular vectors when $p = k$.  The main flaw is that the
noise tolerance depends on the eigengap $\sigma_{k} - \sigma_{k+1}$,
which could be very small.  We have some conjectures for results that
do not depend on this eigengap.

First, when $p > k$, we think that \corollaryref{random} might hold
using the gap $\sigma_{k} - \sigma_{p+1}$ instead of $\sigma_{k} -
\sigma_{k+1}$.  Unfortunately, our proof technique relies on the
principal angle decreasing at each step, which does not necessarily
hold with the larger level of noise.  Nevertheless we expect the
principal angle to decrease fairly fast on average, so that $X_L$ will
contain a subspace very close to $U$.  We are actually unaware of this
sort of result even in the noiseless setting.

\begin{conjecture}
  Let $X_0$ be a random $p$-dimensional basis for $p > k$.  Suppose at
  every step we have
  \[
  100\norm{G_\ell} \leq \eps (\sigma_k - \sigma_{p+1}) \quad\text{and}\quad 100\norm{U^TG_\ell} \leq \frac{\sqrt{p} - \sqrt{k-1}}{\sqrt{d}}
  \]
  Then with high probability, after $L = O(\frac{\sigma_k}{\sigma_k -
    \sigma_{p+1}}\log (d/\eps))$ iterations we have
  \[
  \norm{(I - X_LX_L^\trans)U} \leq \eps.
  \]
\end{conjecture}

The second way of dealing with a small eigengap would be to relax our
goal. \corollaryref{random} is quite stringent in that it requires
$X_L$ to approximate the top $k$ singular vectors $U$, which gets
harder when the eigengap approaches zero and the $k$th through $p+1$st
singular vectors are nearly indistinguishable.  A relaxed goal would
be for $X_L$ to spectrally approximate $A$, that is
\begin{align}\label{eq:Aapprox}
  \norm{(I - X_LX_L^\trans)A} \leq \sigma_{k+1} + \eps.
\end{align}
This weaker goal is known to be achievable in the noiseless setting
without any eigengap at all.  In particular,~\cite{HalkoMT11} shows
that~\eqref{eq:Aapprox} happens after $L = O(\frac{\sigma_{k+1}}{\eps}\log
n)$ steps in the noiseless setting.  A plausible extension to the
noisy setting would be:

\begin{conjecture}
\label{conjecture:spectral}
  Let $X_0$ be a random $2k$-dimensional basis.  Suppose at every step
  we have
  \[
  \norm{G_\ell} \leq \eps \quad\text{and}\quad \norm{U^TG_\ell} \leq \eps \sqrt{k/d}
  \]
  Then with high probability, after $L = O(\frac{\sigma_{k+1}}{\eps}\log
  d)$ iterations we have that
  \[
  \norm{(I - X_LX_L^\trans)A} \leq \sigma_{k+1} + O(\eps).
  \]
\end{conjecture}

\section{Convergence of the noisy power method}
\sectionlabel{robust}

\figureref{NPM} presents our basic algorithm that we analyze
in this section.  An important tool in our analysis are principal
angles, which are useful in analyzing the convergence behavior of
numerical eigenvalue methods. Roughly speaking, we will show that the
tangent of the $k$-th principal angle between $X$ and the top $k$
eigenvectors of $A$ decreases as $\sigma_{k+1}/\sigma_k$ in each
iteration of the noisy power method.

\begin{definition}[Principal angles]
  Let $\cal X$ and $\cal Y$ be subspaces of $\R^d$ of dimension at
  least $k$.  The \emph{principal angles} $0 \leq \theta_1 \leq \dotsb
  \leq \theta_k$ between $\cal X$ and $\cal Y$ and associated
  \emph{principal vectors} $x_1, \dotsc, x_k$ and $y_1, \dotsc, y_k$
  are defined recursively via
  \[
  \theta_i({\cal X}, {\cal Y}) = \min \left\{\arccos \left( \frac{\langle x, y\rangle}{\norm{x}_2\norm{y}_2} \right) \,:\, 
    x \in {\cal X}, y \in {\cal Y},
    x \perp x_j, y \perp y_j \text{ for all } j < i \right\}
  \]
  and $x_i, y_i$ are the $x$ and $y$ that give this value.
  For matrices $X$ and $Y$, we use $\theta_k(X, Y)$ to denote the
  $k$th principal angle between their ranges.
\end{definition}

\subsection{Convergence argument}
We will make use of a non-recursive expression for the principal
angles, defined in terms of the set $\cP_k$ of $p \times p$ projection
matrices $\Pi$ from $p$ dimensions to $k$ dimensional subspaces:

\begin{claim}
  Let $U \in \R^{d \times k}$ have orthonormal columns and $X \in
  \R^{d \times p}$ have independent columns, for $p \geq k$.  Then
  \[
  \cos \theta_k(U, X)
 = \max_{\Pi \in \cP_k} \min_{\substack{x \in
    \range(X\Pi)\\\norm{x}_2 = 1}} \norm{U^\trans x}
 = \max_{\Pi \in \cP_k} \min_{\substack{\norm{w}_2 = 1\\\Pi w = w}} \frac{\norm{U^\trans X w}}{\norm{Xw}}.
  \]
  For $V = U^\perp$, we have
  \[
  \tan \theta_k(U, X) = \min_{\Pi \in \cP_k} \max_{x \in
    \range(X\Pi)} \frac{\norm{V^\trans x}}{\norm{U^\trans x}}
  = \min_{\Pi \in \cP_k} \max_{\substack{\norm{w}_2 = 1\\\Pi w = w}} \frac{\norm{V^\trans X w}}{\norm{U^\trans X w}}.
  \]
\end{claim}

Fix parameters $1\le k\le p\le d.$ In this section we consider a
symmetric $d\times d$ matrix $A$ with singular values $\sigma_1 \geq
\sigma_2 \geq \dotsb \geq \sigma_d$.  We let $U \in \R^{d \times k}$
contain the first $k$ eigenvectors of $A$.  Our main lemma shows that
$\tan \theta_k(U, X)$ decreases multiplicatively in each step.

\begin{lemma}\label{l:decrease}
\lemmalabel{decrease}
  Let $U$ contain the largest $k$ eigenvectors of a symmetric matrix
  $A \in \R^{d \times d}$, and let $X \in \R^{d \times p}$ with $X^trans
X=\mathrm{Id}$ for some $p \geq k$.  Let $G \in \R^{d \times p}$ satisfy
  \begin{align*}
    4 \norm{U^\trans G} &\leq (\sigma_k - \sigma_{k+1}) \cos \theta_k(U, X)\\
    4 \norm{G} &\leq (\sigma_k - \sigma_{k+1}) \eps.
  \end{align*}
  for some $\eps < 1$.  Then
  \[
  \tan \theta_k(U, AX + G) \leq \max\left(\eps, \max\left(\eps, \left(\frac{\sigma_{k+1}}{\sigma_k}\right)^{1/4}\right)\tan \theta_k(U, X)\right).
  \]
\end{lemma}

\begin{proof}
    Let $\Pi^*$ be the matrix projecting onto the smallest $k$
    principal angles of $X$, so that
    \[
    \tan \theta_k(U, X) = \max_{\substack{\norm{w}_2 = 1\\\Pi^* w = w}} \frac{\norm{V^\trans Xw}}{\norm{U^\trans Xw}}.
    \]
    We have that
  \begin{align}
    \tan \theta_k(U, AX+G) &= \min_{\Pi \in \cP_k} \max_{\substack{\norm{w}_2 = 1\\\Pi w = w}} \frac{\norm{V^\trans (AX+G)w}}{\norm{U^\trans (AX+G)w}}\notag\\
    &\leq \max_{\substack{\norm{w}_2 = 1\\\Pi^* w = w}} \frac{\norm{V^\trans AXw} + \norm{V^\trans Gw}}{\norm{U^\trans AXw} - \norm{U^\trans Gw}}\notag\\
    &\leq \max_{\substack{\norm{w}_2 = 1\\\Pi^* w = w}} \frac{1}{\norm{U^\trans Xw}} \cdot \frac{\sigma_{k+1}\norm{V^\trans Xw} + \norm{V^\trans Gw}}{\sigma_k - \norm{U^\trans Gw}/\norm{U^\trans Xw}}\label{eq:3}
  \end{align}
  Define $\Delta = (\sigma_k - \sigma_{k+1}) / 4$.  By the assumption
  on $G$,
  \[
  \max_{\substack{\norm{w}_2 = 1\\\Pi^* w = w}} \frac{\norm{U^\trans Gw}}{\norm{U^\trans Xw}} \leq \norm{U^\trans G} / \cos \theta_k(U, X)  \leq (\sigma_k - \sigma_{k+1}) / 4 = \Delta.
  \]
  Similarly, and using that $1/\cos \theta \leq 1 + \tan \theta$ for
  any angle $\theta$,
  \[
  \max_{\substack{\norm{w}_2 = 1\\\Pi^* w = w}} \frac{\norm{V^\trans Gw}}{\norm{U^\trans Xw}} \leq \norm{G} / \cos \theta_k(U, X) \leq \eps \Delta (1 + \tan \theta_k(U, X)).
  \]
  Plugging back into~\eqref{eq:3} and using $\sigma_k = \sigma_{k+1} + 4\Delta$,
  \begin{align*}
    \tan \theta_k(U, AX+G) &\leq
\max_{\substack{\norm{w}_2 = 1\\\Pi^* w = w}} \frac{\norm{V^\trans Xw}}{\norm{U^\trans Xw}} \cdot \frac{\sigma_{k+1} }{\sigma_{k+1} + 3 \Delta} + \frac{\eps \Delta  (1 + \tan \theta_k(U, X))}{\sigma_{k+1}+3\Delta}.\\
& = \frac{\sigma_{k+1} + \eps \Delta}{\sigma_{k+1} + 3\Delta} \tan \theta_k(U, X) + \frac{\eps \Delta}{\sigma_{k+1} + 3 \Delta}\\
& = (1 - \frac{\Delta}{\sigma_{k+1} + 3\Delta}) \frac{\sigma_{k+1} + \eps \Delta}{\sigma_{k+1} + 2\Delta} \tan \theta_k(U, X) + \frac{\Delta}{\sigma_{k+1} + 3 \Delta} \eps\\
&\leq \max(\eps, \frac{\sigma_{k+1} + \eps \Delta}{\sigma_{k+1} + 2\Delta} \tan \theta_k(U, X))
  \end{align*}
  where the last inequality uses that the weighted mean of two terms
  is less than their maximum.  Finally, we have that
  \[
  \frac{\sigma_{k+1} + \eps \Delta}{\sigma_{k+1} + 2\Delta} \leq \max ( \frac{\sigma_{k+1}}{\sigma_{k+1} + \Delta}, \eps)
  \]
  because the left hand side is a weighted mean of the components on
  the right.  Since $\frac{\sigma_{k+1}}{\sigma_{k+1} + \Delta} \leq
  (\frac{\sigma_{k+1}}{\sigma_{k+1} + 4\Delta} )^{1/4} =
  (\sigma_{k+1}/\sigma_k)^{1/4}$, this gives the result.
\end{proof}

We can inductively apply the previous lemma to get the following general
convergence result.

\begin{theorem}
  \theoremlabel{convergence} Let $U$ represent the top $k$
  eigenvectors of the matrix $A$ and $\gamma = 1 -
  \sigma_{k+1}/\sigma_k$.  Suppose that the initial subspace $X_0$ and
  noise $G_\ell$ is such that
  \begin{align*}
    5\norm{U^\trans G_\ell} &\leq (\sigma_k - \sigma_{k+1}) \cos \theta_k(U, X_0)\\
    5\norm{G_\ell} &\leq \eps (\sigma_k - \sigma_{k+1})
  \end{align*}
  at every stage $\ell$, for some $\eps < 1/2$.  Then there exists an
  $L \lesssim \frac{1}{\gamma}\log\left(\frac{\tan \theta_k(U,
X_0)}{\eps}\right)$
  such that for all $\ell \geq L$ we have $\tan \theta(U, X_L) \leq
  \eps$.
\end{theorem}

\begin{proof}[Proof of \theoremref{convergence}]
  We will see that at every stage $\ell$ of the algorithm,
  \[
  \tan \theta_k(U, X_\ell) \leq \max(\eps, \tan \theta_k(U, X_0))
  \]
  which implies for $\eps \leq 1/2$ that
  \[
  \cos \theta_k(U, X_\ell) \geq \min(1 - \eps^2/2, \cos \theta_k(U, X_0)) \geq \frac{7}{8}\cos \theta_k(U, X_0)
  \]
  so Lemma~\ref{l:decrease} applies at every stage.  This means that
  \[
  \tan \theta_k(U, X_{\ell+1}) = \tan \theta_k(U, AX_{\ell} + G) \leq \max(\eps,
  \delta \tan \theta_k(U, X_{\ell}))
  \]
  for $\delta = \max(\eps, (\sigma_{k+1}/\sigma_k)^{1/4})$.  After
  \[
  L = \log_{1/\delta} \frac{\tan \theta_k(U, X_0)}{\eps}
  \]
  iterations the tangent will reach $\eps$ and remain there.
  Observing that
  \[
  \log (1/\delta) \gtrsim \min(\log (1/\eps), \log
  (\sigma_k/\sigma_{k+1})) \geq \min(1, \log \frac{1}{1 - \gamma}) \geq \min(1, \gamma) = \gamma
  \]
  gives the result.
\end{proof}

\subsection{Random initialization}

The next lemma essentially follows from bounds on the smallest singular value
of gaussian random matrices~\cite{RV09}.

\begin{lemma}\label{c:randomsubspace}
\lemmalabel{randomsubspace}
For an arbitrary orthonormal $U$ and random subspace
  $X$, we have
  \[
  \tan \theta_k(U, X) \leq \tau \frac{\sqrt{d}}{\sqrt{p} - \sqrt{k-1}}
  \]
  with all but $\tau^{-\Omega(p + 1 - k)} + e^{-\Omega(d)}$ probability.
\end{lemma}

\begin{proof}
    Consider the singular value decomposition $U^\trans X = A \Sigma B^\trans $ of
    $U^\trans  X$.  Setting $\Pi$ to be matrix projecting onto the first $k$
    columns of $B$, we have that
    \[
    \tan \theta_k(U, X) \leq \max_{\substack{\norm{w}_2 = 1\\\Pi w  = w}} \frac{\norm{V^\trans Xw}}{\norm{U^\trans Xw}} \leq \norm{V^\trans X} \max_{\substack{\norm{w}_2 = 1\\\Pi w  = w}} \frac{1}{\norm{\Sigma B^\trans  w}} = \norm{V^\trans X} \max_{\substack{\norm{w}_2 = 1\\\supp(w) \in [k]}} \frac{1}{\norm{\Sigma w}} = \frac{\norm{V^\trans X}}{\sigma_k(U^\trans X)}.
    \]

    Let $X \sim N(0, I_{d \times p})$ represent the random subspace.
    Then $Y := U^\trans X \sim N(0, I_{k \times p})$.  By~\cite{RV09},
    for any $\eps$, the smallest singular value of $Y$ is at least
    $(\sqrt{p} - \sqrt{k-1})/\tau$ with all but $\tau^{-\Omega(p + 1 -
      k)} + e^{-\Omega(p)}$ probability.  On the other hand, $\norm{X}
    \lesssim \sqrt{d}$ with all but $e^{-\Omega(d)}$ probability.
    Hence
    \[
    \tan \theta_k(U, X) \lesssim \tau \frac{\sqrt{d}}{\sqrt{p} - \sqrt{k-1}}
    \]
    with the desired probability.  Rescaling $\tau$ gets the result.
\end{proof}

With this lemma we can prove the corollary that we stated in the introduction.

\begin{proof}[Proof of \corollaryref{random}]
  By Claim~\ref{c:randomsubspace}, with the desired probability we
  have $ \tan \theta_k(U, X_0) \leq \frac{\tau \sqrt{d}}{\sqrt{p} -
    \sqrt{k-1}}.  $ Hence $\cos \theta_k(U, X_0) \geq 1/(1 + \tan
  \theta_k(U, X_0)) \geq \frac{\sqrt{p} - \sqrt{k-1}}{2 \cdot \tau
    \sqrt{d}}$.  Rescale $\tau$ and apply
  Theorem~\ref{thm:convergence} to get that $\tan \theta_k(U, X_L)
  \leq \eps$.  Then
  $\norm{(I - X_L X_L^\trans)U} = \sin \theta_k(U, X_L) \leq \tan \theta_k(U, X_L)
  \leq \eps.$
\end{proof}

\section{Memory efficient streaming PCA}
\sectionlabel{streaming}

In the streaming PCA setting we receive a stream of samples
$z_1,z_2,\dots\in\R^d.$ Each sample is drawn i.i.d.~from an unknown
distribution~$\cD$ over $\R^d.$ Our goal is to compute the dominant
$k$ eigenvectors of the covariance matrix $A=\E_{z\sim\cD} zz^\trans.$
The challenge is to do this with small space, so we cannot store the
$d^2$ entries of the sample covariance matrix.  We would like to use
$O(dk)$ space, which is necessary even to store the output.
%
%

The streaming power method~(Figure~\ref{fig:SPM}, introduced
by~\cite{MCJ13}) is a natural algorithm that performs streaming PCA
with $O(dk)$ space.  The question that arises is how many samples it
requires to achieve a given level of accuracy, for various
distributions $\cD$.  Using our general analysis of the noisy power
method, we show that the streaming power method requires fewer samples
and applies to more distributions than was previously known.

We analyze a broad class of distributions:
\begin{definition}
A distribution $\cD$ over $\R^d$ is \emph{$(B,p)$-round} if for every
$p$-dimensional projection $P$ and all $t\ge 1$ we have
$\Pr_{z\sim\cD}\Set{\|z\|>t}\le\exp(-t)$ and $\Pr_{z\sim\cD}\Set{\|Pz\| >
t\cdot \sqrt{Bp/d}}\le
\exp(-t)\mper$
\end{definition}
The first condition just corresponds to a normalization of the samples
drawn from~$\cD.$ Assuming the first condition holds, the second
condition always holds with $B=d/p.$ For this reason our analysis in
principle applies to any distribution, but the sample complexity will
depend quadratically on $B$.

Let us illustrate this definition through the example of the spiked
covariance model studied by~\cite{MCJ13}. The spiked covariance model
is defined by an orthonormal basis $U\in\R^{d\times k}$ and a diagonal
matrix $\Lambda\in\R^{k\times k}$ with diagonal entries
$\lambda_1\ge\lambda_2\ge\cdots\ge\lambda_k>0.$ The distribution
$\cD(U,\Lambda)$ is defined as the normal distribution
$\mathrm{N}(0,(U\Lambda^2 U^\trans + \sigma^2\mathrm{Id}_{d\times
  d})/D)$ where $D=\Theta(d\sigma^2+\sum_i\lambda_i^2)$ is a
normalization factor chosen so that the distribution satisfies the
norm bound. Note that the the $i$-th eigenvalue of the covariance
matrix is $\sigma_i=(\lambda_i^2+\sigma^2)/D$ for $1\le i\le k$ and
$\sigma_i=\sigma^2/D$ for $i>k.$ We show in \lemmaref{spiked-round}
that the spiked covariance model $\cD(U,\Lambda)$ is indeed
$(B,p)$-round for $B=O(\frac{\lambda_1^2+\sigma^2}{\tr(\Lambda)/d +
  \sigma^2})$, which is constant for $\sigma \gtrsim \lambda_1$.

We have the following main theorem.
\begin{theorem}\label{thm:streamingpca}
  Let $\cD$ be a $(B,p)$-round distribution over~$\R^d$ with
  covariance matrix~$A$ whose eigenvalues are
  $\sigma_1\ge\sigma_2\ge\cdots\ge\sigma_d\ge0.$ Let $U\in\R^{d\times
    k}$ be an orthonormal basis for the eigenvectors corresponding to
  the first $k$ eigenvalues of $A.$ Then, the streaming power method
  \SPM returns an orthonormal basis $X\in\R^{d\times p}$ such that
  $\tan\theta(U,X)\le\epsilon$ with probability $9/10$ provided that
  \SPM receives $n$ samples from~$\cD$ for some $n$ satisfying
  \[
  n \le \tilde O\left( \frac{B^2\sigma_k k\log^2 d}
    {\eps^2(\sigma_k-\sigma_{k+1})^3d}\right)
  \]
  if $p = k + \Theta(k)$.  More generally, for all $p \geq k$ one can
  get the slightly stronger result
  \[
  n \le \tilde O\left( \frac{Bp\sigma_k\max\{1/\epsilon^2,Bp/(\sqrt{p}-\sqrt{k-1})^2\}\log^2 d}
    {(\sigma_k-\sigma_{k+1})^3d}\right)\mper
  \]
\end{theorem}

Instantiating with the spiked covariance model gives the following:

\begin{corollary}\corollarylabel{spiked}
In the spiked covariance model $\cD(U,\Lambda)$ 
the conclusion of \theoremref{streamingpca} holds 
for $p = 2k$ with
\[
n=\tilde O\left(
\frac{(\lambda_1^2+\sigma^2)^2(\lambda_k^2+\sigma^2)}
{\eps^2\lambda_k^6} dk
\right)\mper
\]
\end{corollary}
When $\lambda_1=O(1)$ and $\lambda_k=\Omega(1)$ this becomes
$n = \tilde O\left(\frac{\sigma^6 + 1}{\eps^2}\cdot dk\right)\mper$

We can apply \theoremref{streamingpca} to all distributions that have
exponentially concentrated norm by setting $B = d/p$.  This gives the
following result.

\begin{corollary}\corollarylabel{spiked}
  Let $\cD$ be any distribution scaled such that $\Pr_{z \sim
    \cD}[\norm{z} > t] \leq \exp(-t)$ for all $t \geq 1$.  Then the
  conclusion of \theoremref{streamingpca} holds for $p=2k$ with
  \[
  n=\tilde O\left( \frac{\sigma_k}{\eps^2 k(\sigma_k - \sigma_{k+1})^3} \cdot
d\right)\mper
  \]
\end{corollary}

If the eigenvalues follow a power law, $\sigma_j \approx j^{-c}$ for a
constant $c > 1/2$, this gives an $n = \tilde O(k^{2c+2}d/\eps^2)$
bound on the sample complexity.

\subsection{Error term analysis}

Fix an orthonormal basis $X\in\R^{d\times k}.$ Let $z_1,\dots,z_n\sim\cD$ be
samples from a distribution~$\cD$ with covariance matrix $A$ and consider the matrix
\[
G = \left(A - \hat A \right)X\mcom
\]
where $\hat A = \frac{1}{n}\sum_{i=1}^n z_i z_i^\trans$ is the empirical covariance
matrix on $n$ samples. Then, we have that 
$\hat A X = AX + G.$
In other words, one update step of the power method executed on $\hat A$ can
be expressed as an update step on $A$ with noise matrix~$G.$ This simple
observation allows us to apply our analysis of the noisy power method to this
setting after obtaining suitable bounds on $\|G\|$ and $\|U^\trans G\|.$

\begin{lemma}\lemmalabel{streaming-error}
  Let $\cD$ be a $(B,p)$-round distribution with covariance
  matrix~$M$. Then with all but $O(1/n^2)$ probability,
\[
\norm{G} \lesssim \sqrt{\frac{Bp \log^4n \log d}{d n}} + \frac{1}{n^2}
\quad\text{and}\quad
\norm{U^\trans G} \lesssim \sqrt{\frac{B^2p^2 \log^4n \log d}{d^2 n}} + \frac{1}{n^2}
\]
\end{lemma}
\begin{proof}
We will use a matrix Chernoff bound to show that
\begin{enumerate}
\item $\Pr\Set{\|G\| > C t\log(n)^2\sqrt{Bp/d}+O(1/n^2)}\le d\exp(-t^2 n) + 1/n^2$
\item $\Pr\Set{\|U^\trans G\| > C t\log(n)^2 Bp/d + O(1/n^2)}\le d\exp(-t^2 n) + 1/n^2$
\end{enumerate}
setting $t = \sqrt{\frac{2}{n}\log d}$ gives the result.  However,
matrix Chernoff inequality requires the distribution to satisfy a norm
bound with probability~$1.$ We will therefore create a closely related
distribution~$\tilde\cD$ that satisfies such a norm constraint and is
statistically indistinguishable up to small error on~$n$ samples. We
can then work with $\tilde\cD$ instead of $\cD.$ This truncation step
is standard and works because of the concentration properties
of~$\cD.$

Indeed, let $\tilde \cD$ be the distribution obtained from $\cD$ be replacing 
a sample $z$ with $0$ if 
\[
\|z\|> C\log(n)\quad\text{ or }\quad\|U^\trans z\|
\ge C\log(n)\sqrt{Bp/d}\quad\text{ or }\quad\|z^\trans X\|>
C\log(n)\sqrt{Bp/d}\mper
\]
For sufficiently large constant $C,$ it follows from the definition of $(B,p)$-round that
the probability that one or more of $n$ samples from $\cD$ get zeroed out is
at most $1/n^2.$ In particular, the two product distributions $\cD^{(n)}$ and
$\tilde \cD^{(n)}$ have total variation distance at most $1/n^2.$
Furthermore, we claim that the covariance matrices of the two distributions
are at most $O(1/n^2)$ apart in spectral norm. 
Formally,
\[
\Norm{\E_{z\sim\cD}zz^\trans -\E_{\tilde z\sim\tilde\cD}\tilde z\tilde
z^\trans} \le \frac1{n^2}\cdot O\left(\int_{t\ge1}
C^2t^2\log^2(n)\exp(-t)\mathrm{d}t \right)
\le O(1/n^2)
\mper
\]
In the first inequality we use the fact that $z$ only gets zeroed out
with probability $1/n^2.$ Conditional on this event, the norm of $z$
is larger than $tC\log(n)$ with probability at most
$n^2\exp(-\frac{1}{2}tC\log n) \leq \exp(-t).$ Assuming the norm is at
most $tC\log(n)$ we have $\Norm{zz^\trans}\le t^2C^2\log^2(n)$ and
this bounds the contribution to the spectral norm of the difference.

Now let $\tilde G$ be the error matrix defined as $G$ except that we replace
the samples $z_1,\dots,z_n$ by $n$ samples~$\tilde z_1,\dots,\tilde z_n$ from
the truncated distribution~$\tilde\cD.$ 
By our preceding discussion, it now suffices to show that 
\begin{enumerate}
\item $\Pr\Set{\|\tilde G\| > C t\log^2(n)\sqrt{Bp/d}}\le d\exp(-t^2 n)$
\item $\Pr\Set{\|U^\trans \tilde G\| > C t\log^2(n) Bp/d}\le d\exp(-t^2 n)$
\end{enumerate}
To see this, let $S_i = \tilde z_i\tilde z_i^\trans X.$ We have
\[
\Norm{S_i} \le \|\tilde z_i\|\cdot\Norm{\tilde z_i^\trans X} \le
C^2\log^2(n)\cdot\sqrt{Bp/d}
\]
Similarly,
\[
\Norm{U^\trans S_i} \le \|U^\trans \tilde z_i\|\cdot\Norm{\tilde z_i^\trans X} \le
C^2\log^2(n)\cdot \frac{Bp}{d}\mper
\]
The claims now follow directly from the matrix Chernoff bound
stated in~\lemmaref{matrixb2}.
\end{proof}

\subsection{Proof of \theoremref{streamingpca}} 
Given \lemmaref{streaming-error} we will choose $n$ such that the
error term in each iteration satisfies the assumptions of
\theoremref{convergence}. Let $G_\ell$ denote the instance of the
error term~$G$ arising in the $\ell$-th iteration of the algorithm. We
can find an $n$ satisfying
\[
\frac{n}{\log(n)^4} =
O\left(\frac{Bp\max\Set{1/\epsilon^2,Bp/(\sqrt{p}-\sqrt{k-1})^2}\log d}{(\sigma_k-\sigma_{k+1})^2d}\right)
\]
such that by \lemmaref{streaming-error} we have that with probability
$1-O(1/n^2),$
\[
\|G_\ell\|\le\frac{\epsilon(\sigma_k-\sigma_{k+1})}{5}
\quad\text{and}\quad
\|U^\trans G_\ell\|\le \frac{\sigma_k-\sigma_{k+1}}{5} \frac{\sqrt{p} - \sqrt{k-1}}{\sqrt{d}}\mper
\]
Here we used that by definition $1/n \ll \epsilon$ and $1/n \ll
\sigma_k-\sigma_{k+1}$ and so the $1/n^2$ term in \lemmaref{streaming-error}
is of lower order.

With this bound, it follows from \theoremref{convergence} that after
$L=O(\log(d/\epsilon)/(1-\sigma_{k+1}/\sigma_k))$ iterations we have
with probability $1-\max\{1,L/n^2\}$ that
$\tan\theta(U,X_L)\le\epsilon.$ The over all sample complexity is therefore
\[
Ln = 
\tilde O\left(\frac{Bp\sigma_k\max\Set{1/\epsilon^2,Bp/(\sqrt{p}-\sqrt{k-1})^2}\log^2 d}{(\sigma_k-\sigma_{k+1})^3d}\right)
\mper
\]
Here we used that $1-\sigma_{k+1}/\sigma_k =
(\sigma_k-\sigma_{k+1})/\sigma_k.$ This concludes the proof of
\theoremref{streamingpca}.

\subsection{Proof of \lemmaref{spiked-round} and \corollaryref{spiked}}

\begin{lemma}\lemmalabel{spiked-round}
The spiked covariance model $\cD(U,\Lambda)$ is $(B,k)$-round for
$B=O(\frac{\lambda_1^2+\sigma^2}{\tr(\Lambda)/d + \sigma^2}).$
\end{lemma}
\begin{proof}
Note that an example $z\sim\cD(U,\Lambda)$ is distributed as $U\Lambda g + g'$
where $g\sim \mathrm{N}(0,1/D)^k$ is a standard gaussian and $g'\sim
\mathrm{N}(0,\sigma^2/D)^d.$ is a noise term. Recall, that $D$ is the
normalization term. Let $P$ be any projection operator onto a $k$-dimensional
space. Then,
\[
\|Pz\| = \|PU\Lambda g + Pg'\|
\le \|PU\Lambda g\| + \|Pg'\|
\le \|\Lambda g\| + \|Pg'\|
\le \lambda_1\|g\| + \|Pg'\|
\mper
\]
By rotational invariance of $g'$, we may assume that $P$ is the projection
onto the first $k$ coordinates. Hence, $\|Pg'\|$ is distributed like the norm
of $\mathrm{N}(0,\sigma^2/D)^k.$ 
Using standard tail bounds for the norm of a gaussian random variables, we can see that
$\|Pz\|^2= O(t(k\lambda_1^2 + k\sigma^2)/D)$ with
probability $1-\exp(-t).$ On the other hand,
$D = \Theta(\sum_{i=1}^k\lambda_i^2 + d\sigma^2).$ We can now solve for $B$
by setting 
\[
\Theta(\frac{k\lambda_1^2 + k\sigma^2}{\sum_{i=1}^k\lambda_i^2+d\sigma^2})
= \frac{Bk}{d}
\quad\Leftrightarrow\quad
B = \Theta(\frac{\lambda_1^2 + \sigma^2}{\frac{1}{d}\sum_{i=1}^k\lambda_i^2+\sigma^2})
\mper
\]
\end{proof}
\corollaryref{spiked} follows by plugging in the bound on $B$ and the
eigenvalues of the covariance matrix into our main theorem.
\begin{proof}[Proof of \corollaryref{spiked}]
In the spiked covariance model $\cD(U,\Lambda)$ we have 
\[
B = \frac{\lambda_1^2 + \sigma^2}{D}\mcom
\quad \sigma_k = \frac{\lambda_k^2 + \sigma^2}{D}\mcom
\quad \sigma_{k+1} = \frac{\sigma^2}D\mcom
\quad D = O(\tr(\Lambda^2)+d\sigma^2)\mper
\]
Hence,
\[
\frac{B^2\sigma_k}
{(\sigma_k-\sigma_{k+1})^3 d}
= \frac{(\lambda_1^2+\sigma^2)^2(\lambda_k^2 + \sigma^2)}{\lambda_k^6d}
\le \frac{(\lambda_1^2+\sigma^2)^3}{\lambda_k^6d}
\]
Plugging this bound into \theoremref{streamingpca} gives
\corollaryref{spiked}.
\end{proof}

\section{Privacy-preserving singular vector computation}
\sectionlabel{privacy}

In this section we prove our results about privacy-preserving singular vector
computation. We begin with a standard definition of differential privacy,
sometimes referred to as \emph{entry-level differential privacy}, as it
hides the presence or absence of a single entry.

\begin{definition}[Differential Privacy]
\definitionlabel{dp}
A randomized algorithm $M\colon\mathbb{R}^{d\times d'}\rightarrow R$ (where $R$ is some
arbitrary abstract range) is \emph{$(\epsilon,\delta)$-differentially private}
if for all pairs of matrices $A, A' \in \mathbb{R}^{d\times d'}$ differing in
only one entry by at most $1$ in absolute value, we have that for all subsets of 
the range $S \subseteq R,$ the algorithm satisfies:
$\Pr\Set{M(A) \in S} \leq
\exp(\epsilon)\Pr\Set{M(A') \in S} + \delta\mper$
\end{definition}
The definition is most meaningful when $A$ has entries in $[0,1]$ so that the
above definition allows for a single entry to change arbitrarily within this
range. However, this is not a requirement for us.
The privacy guarantee can be strengthened by decreasing $\epsilon>0.$

For our choice of $\sigma$ in \figureref{PPM} the algorithm satisfies 
$(\epsilon,\delta)$-differential privacy as follows easily
from properties of the Gaussian distribution. See, for example,~\cite{HardtR13} for a proof.
\begin{claim}
\PPM satisfies $(\epsilon,\delta)$-differential privacy.
\end{claim}

It is straightforward to prove \theoremref{privacy-main} by invoking our
convergence analysis of the noisy power method together with suitable error
bounds. The error bounds are readily available as the noise term is just
gaussian.

\begin{proof}[Proof of \theoremref{privacy-main}]
  Let $m = \max\|X_\ell\|_\infty.$ By~\lemmaref{gaussian-projection}
  the following bounds hold with probability $99/100$:
\begin{enumerate}
\item $\max_{\ell=1}^L\|G_\ell\| \lesssim \sigma m \sqrt{d \log L}$
\item $\max_{\ell=1}^L\|U^\trans G_\ell\| \lesssim \sigma m \sqrt{k\log L}$
\end{enumerate}
Let
\[
\eps' = \frac{\sigma m \sqrt{d \log L}}{\sigma_k - \sigma_{k+1}}\gtrsim \frac{5\max_{\ell=1}^L\|G_\ell\|}{\sigma_k - \sigma_{k+1}}.
\]
By \corollaryref{random}, if we also have that
$\max_{\ell=1}^L\|U^\trans G_\ell\| \leq (\sigma_k - \sigma_{k+1})
\frac{\sqrt{p}-\sqrt{k-1}}{\tau \sqrt{d}}$ for a sufficiently large
constant $\tau$, then we will have that
\[
\norm{(I-X_LX_L^\trans)U} \leq \eps'
\leq \frac{\sigma m \sqrt{d \log L}}{\sigma_k - \sigma_{k+1}}
\]
after the desired number of iterations, giving the theorem.
Otherwise, 
\[
(\sigma_k - \sigma_{k+1})\frac{\sqrt{p}-\sqrt{k-1}}{\tau\sqrt{d}} \leq
\max_{\ell=1}^L\|U^\trans G_\ell\| \lesssim \eps' (\sigma_k -
\sigma_{k+1})\sqrt{k/d},
\]
so it is trivially true that
\[
\frac{\sigma m \sqrt{d \log L}}{\sigma_k - \sigma_{k+1}} \frac{\sqrt{p}}{\sqrt{p} - \sqrt{k-1}} \geq \eps' \frac{\sqrt{k}}{\sqrt{p} - \sqrt{k-1}} \gtrsim 1 \geq \norm{(I-X_LX_L^\trans)U}.
\]

\end{proof}

\subsection{Low-rank approximation}

Our results readily imply that we can compute accurate differentially private
low-rank approximations. The main observation is that, assuming $X_L$ and $U$
have the same dimension, $\tan\theta(U,X_L)\le\alpha$ implies that the matrix $X_L$ 
also leads to a good
low-rank approximation for $A$ in the spectral norm. In particular
\begin{equation}\equationlabel{low-rank}
\|(I-X_LX_L^\trans)A\|\le \sigma_{k+1} + \alpha\sigma_1\mper
\end{equation}
Moreover the projection step of computing $X_LX_L^\trans A$ can be carried out
easily in a privacy-preserving manner. It is again the $\ell_\infty$-norm of
the columns of $X_L$ that determine the magnitude of noise that is needed.
Since $A$ is symmetric, we have $X^\trans A=(AX)^\trans.$ Hence, to obtain a
good low-rank approximation it suffices to compute the product $AX_L$
privately as $AX_L + G_L.$ This leads to the following corollary.

\begin{corollary}\corollarylabel{privacy-lowrank}
Let $A\in\R^{d\times d}$ be a symmetric matrix with singular values
$\sigma_1\ge\dots\ge\sigma_d$ and let $\gamma=1-\sigma_{k+1}/\sigma_{k}.$ 
There is an $(\epsilon,\delta)$-differentially private
algorithm that given $A$ and $k,$ outputs a rank $2k$ matrix $B$ such that with
probability $9/10,$
\[
\|A-B\| \le \sigma_{k+1}
+
\tilde O\left( 
\frac{\sigma_1\sqrt{(k/\gamma)d\log
d\log(1/\delta)}}{\epsilon(\sigma_k-\sigma_{k+1})}
\right)\mper
\]
The $\tilde O$-notation hides the factor
$O\big(\sqrt{\log(\log(d)/\gamma)}\big).$
\end{corollary}
\begin{proof}
Apply \theoremref{privacy-main} with $p=2k$ and run the algorithm for $L+1$
steps with $L=O(\gamma^{-1}\log d).$ This gives the bound
\[
\alpha = \norm{(I - X_LX_L^\trans)A} \le O\left(\frac{\sqrt{(k/\gamma)d\log
d\log(\log(d)/\gamma)\log(1/\delta)}}{\epsilon(\sigma_k-\sigma_{k+1})}\right)\mper
\]
Moreover, the algorithm has computed $Y_{L+1} = AX_L+G_L$ and we have
$B =  X_LY_{L+1}^\trans = X_LX_L^\trans A +X_LG_L^\trans.$ Therefore
\[
\Norm{A-B} \le \sigma_{k+1} + \alpha\sigma_1 + \Norm{X_LG_L^\trans}
\]
where $\Norm{X_LG_L^\trans}\le\Norm{G_L}.$ 
By definition of the algorithm and \lemmaref{gaussian-projection}, we have
\[
\Norm{G_L}
\le O\left(\sqrt{\sigma^2d}\right)
= O\left(\frac1{\epsilon}\sqrt{(k/\gamma)d\log(d)\log(1/\delta)}\right)\mper
\]
Given that the $\alpha$-term gets multiplied by $\sigma_1,$ this bound on
$\Norm{G_L}$ is of lower order and the corollary follows.
\end{proof}

\subsection{Principal Component Analysis}
\sectionlabel{PCA}

Here we illustrate that our bounds directly imply results for the
privacy notion studied by Kapralov and Talwar~\cite{KapralovT13}. The
notion is particularly relevant in a setting where we think of $A$ as
a sum of rank~$1$ matrices each of bounded spectral norm.

\begin{definition}
\definitionlabel{dp-spectral}
A randomized algorithm $M\colon\mathbb{R}^{d\times d'}\rightarrow R$ (where $R$ is some
arbitrary abstract range) is \emph{$(\epsilon,\delta)$-differentially private
under unit spectral norm changes}
if for all pairs of matrices $A, A' \in \mathbb{R}^{d\times d'}$ 
satisfying $\|A-A'\|_2\le1,$ we have that for all subsets of 
the range $S \subseteq R,$ the algorithm satisfies:
$\Pr\Set{M(A) \in S} \leq
\exp(\epsilon)\Pr\Set{M(A') \in S} + \delta\mper$
\end{definition}

\begin{lemma}\lemmalabel{privacy-spectral}
If \PPM is executed with each $G_\ell$ sampled independently as $G_\ell \sim
N(0,\sigma^2)^{d\times p}$ with $\sigma
=\epsilon^{-1}\sqrt{4pL\log(1/\delta)},$ then \PPM satisfies
$(\epsilon,\delta)$-differential privacy under unit spectral norm changes.

If $G_\ell$ is sampled with i.i.d. Laplacian entries 
$G_\ell\sim \mathrm{Lap}(0,\lambda)^{n\times k}$ where $\lambda =
10\epsilon^{-1}pL\sqrt{d},$ then \PPM satisfies $(\epsilon,0)$-differential
privacy under unit spectral norm changes.
\end{lemma}
\begin{proof}
The first claim follows from the privacy proof in~\cite{HardtR12}. We sketch 
the argument here for completeness.
Let $D$ be any matrix with $\|D\|_2\le1$ (thought of as $A-A'$ in
\definitionref{dp-spectral}) and let $\|x\|=1$ be any unit vector which we
think of as one of the columns of $X=X_{\ell-1}.$ Then, we have $\|Dx\|\le
\|D\|\cdot\|x\|\le 1,$ by definition of the spectral norm. This shows that the
``$\ell_2$-sensitivity'' of one matrix-vector multiplication in our algorithm
is bounded by~$1.$ It is well-known that it suffices to add Gaussian noise
scaled to the $\ell_2$-sensitivity of the matrix-vector product in order to
achieve differential privacy. Since there are $kL$ matrix-vector
multiplications in total we need to scale the noise by a factor
of~$\sqrt{kL}.$

The second claim follows analogously. Here however we need to scale the noise
magnitude to the ``$\ell_1$-sensitivity'' of the matrix-vector product which
be bound by $\sqrt{n}$ using Cauchy-Schwarz. The claim then follows using
standard properties of the Laplacian mechanism.
\end{proof}

Given the previous lemma it is straightforward to derive the following
corollaries.
\begin{corollary}
Let $A\in\R^{d\times d}$ be a symmetric matrix with singular values
$\sigma_1\ge\dots\ge\sigma_d$ and let $\gamma=1-\sigma_{k+1}/\sigma_{k}.$ 
There is an algorithm that given
a $A$ and parameter $k,$ preserves
$(\epsilon,\delta)$-differentially privacy under unit spectral norm changes
and outputs a rank $2k$ matrix $B$ such that with
probability $9/10,$
\[
\|A-B\| \le \sigma_{k+1}
+
\tilde O\left( 
\frac{\sigma_1\sqrt{(k/\gamma)d\log
d\log(1/\delta)}}{\epsilon(\sigma_k-\sigma_{k+1})}
\right)\mper
\]
The $\tilde O$-notation hides the factor
$O\big(\sqrt{\log(\log(d)/\gamma)}\big).$
\end{corollary}
\begin{proof}
The proof is analogous to the proof of \corollaryref{privacy-lowrank}.
\end{proof}
A similar corollary applies to $(\epsilon,0)$-differential privacy. 
\begin{corollary}
Let $A\in\R^{d\times d}$ be a symmetric matrix with singular values
$\sigma_1\ge\dots\ge\sigma_d$ and let $\gamma=1-\sigma_{k+1}/\sigma_{k}.$ 
There is an algorithm that given
a $A$ and parameter $k,$ preserves
$(\epsilon,\delta)$-differentially privacy under unit spectral norm changes
and outputs a rank $2k$ matrix $B$ such that with
probability $9/10,$
\[
\|A-B\| \le \sigma_{k+1}
+
\tilde O\left( 
\frac{\sigma_1 k^{1.5}d\log(d)\log(d/\gamma)}
{\epsilon\gamma(\sigma_k-\sigma_{k+1})}
\right)\mper
\]
\end{corollary}

\begin{proof}
We invoke \PPM with $p=2k$ and Laplacian noise with
the scaling given by \lemmaref{privacy-spectral} so that the algorithm
satisfies $(\epsilon,0)$-differential privacy.
Specifically,
$G_\ell\sim \mathrm{Lap}(0,\lambda)^{d\times p}$ where $\lambda =
10\epsilon^{-1}pL\sqrt{d}.$ 
\lemmaref{laplacian-projection}. Indeed, with probability $99/100,$ we have
\begin{enumerate}
\item $\max_{\ell=1}^L\|G_\ell\| \le O\left(\lambda\sqrt{kd}\log(kdL)\right)
= O\left((1/\epsilon\gamma)k^{1.5}d\log(d)\log(kdL)\right)$
\item $\max_{\ell=1}^L\|U^\trans G_\ell\| \le O\left(\lambda k\log(kL)\right)
= O\left((1/\epsilon\gamma)k^2\sqrt{d}\log(d)\log(kL)\right)$
\end{enumerate}
We can now plug these error bounds into \corollaryref{random} to obtain the
bound
\[
\Norm{(I-X_LX_L^\trans)U}
\le O\left(
\frac{k^{1.5}d\log(d)\log(d/\gamma)}
{\epsilon\gamma(\sigma_k-\sigma_{k+1})}
\right)
\]
Repeating the argument from the proof of \corollaryref{privacy-lowrank} gives
the stated guarantee for low-rank approximation.
\end{proof}

The bound above matches a lower bound shown by Kapralov and
Talwar~\cite{KapralovT13} up to a factor of~$\tilde O(\sqrt{k}).$ We believe
that this factor can be eliminated from our bounds by using a quantitatively
stronger version of~\lemmaref{laplacian-projection}. Compared to the upper
bound of~\cite{KapralovT13} our algorithm is faster by a more than a quadratic
factor in~$d.$ Moreover, previously only bounds for
$(\epsilon,0)$-differential privacy were known for the spectral norm privacy
notion, whereas our bounds strongly improve when going to
$(\epsilon,\delta)$-differential privacy.

\subsection{Dimension-free bounds for incoherent matrices}
\sectionlabel{privacy-mu}

The guarantee in \theoremref{privacy-main} depends on the quantity
$\|X_{\ell}\|_\infty$ which could in principle be as small as $\sqrt{1/d}.$ Yet, in
the above theorems, we use the trivial upper bound~$1.$ This in turn resulted
in a dependence on the dimensions of $A$ in our theorems. Here, we show that
the dependence on the dimension can be replaced by an essentially tight
dependence on the \emph{coherence} of the input matrix. In doing so, we resolve
the main open problem left open by Hardt and Roth~\cite{HardtR13}.
The definition of coherence that we will use is formally defined as
follows.

\begin{definition}[Matrix Coherence]
We say that a matrix $A \in \mathbb{R}^{d\times d'}$ with singular value
decomposition $A = U\Sigma V^\trans$ has \emph{coherence}
\[
\mu(A) \defeq\left\{d\|U\|^2_\infty, d'\|V\|^2_\infty\right\}\mper
\]
Here $\|U\|_\infty=\max_{ij}|U_{ij}|$ denotes the largest entry of $U$ in absolute
value.
\end{definition}
Our goal is to show that the $\ell_\infty$-norm of the vectors
arising in \PPM is closely related to the coherence of the input matrix.
We obtain a nearly tight connection between the coherence of the matrix and the
$\ell_\infty$-norm of the vectors that \PPM computes. 

\begin{theorem}\theoremlabel{gaussian-noise}
Let $A\in\R^{d\times d}$ be symmetric. Suppose \NPM is invoked on $A,$ and
$L\le n,$ with each $G_\ell$ sampled from $N(0,\sigma_\ell^2)^{d\times p}$ for
some $\sigma_\ell>0.$ Then, with probability $1-1/n,$ 
\[
\max_{\ell=1}^L \|X_\ell\|_\infty^2
\le O\left(\frac{\mu(A)\log(d)}d\right)\mper
\]
\end{theorem}

\begin{proof}
Fix $\ell\in[L].$
Let $A=\sum_{i=1}^n\sigma_i u_i u_i^\trans$ be 
given in its eigendecomposition. Note that
\[
B=\max_{i=1}^d\|u_i\|_\infty\le\sqrt{\frac{\mu(A)}d}.
\]
We may write any column $x$ of $X_\ell$ as 
$x=\sum_{i=1}^ds_i\alpha_iu_i$ where
$\alpha_i$ are non-negative scalars such that $\sum_{i=1}^d\alpha_i^2=1,$ and
$s_i\in\signs$ where $s_i=\sign(\langle x,u_i\rangle).$ Hence, by
\lemmaref{signs} (shown below), the signs $(s_1,\dots,s_d)$ are distributed uniformly at
random in $\signs^d.$ Hence, by \lemmaref{sign-deviation} 
(shown below), it follows that
$
\Pr\Set{ \left\|x\right\|_\infty > 4B\sqrt{\log d}} \le 1/n^3\mper
$
By a union bound over all $p\le d$ columns it follows that
$\Pr\Set{ \left\|X_\ell\right\|_\infty > 4B\sqrt{\log d}} \le 1/d^2\mper$
Another union bound over all $L\le d$ steps completes the proof.
\end{proof}

The previous theorem states that no matter what the scaling of the Gaussian
noise is in each step of the algorithm, so long as it is Gaussian the
algorithm will maintain that $X_\ell$ has small coordinates. We cannot hope to
have coordinates smaller than $\sqrt{\mu(A)/d},$ since eventually the
algorithm will ideally converge to $U.$
This result directly implies the theorem we stated in the
introduction.

\begin{proof}[Proof of \theoremref{privacy-mu}]
The claim follows directly from \theoremref{privacy-main} after applying 
\theoremref{gaussian-noise} which shows that with probability~$1-1/n,$
\[
\max_{\ell=1}^L \|X_\ell\|_\infty^2
\le O\left(\frac{\mu(A)\log(d)}d\right)\mper \qedhere
\]
\end{proof}

\subsection{Proofs of supporting lemmas}

We will now establish \lemmaref{signs} and \lemmaref{sign-deviation} that were
needed in the proof of the previous theorem. For that purpose we need some
basic symmetry properties of the QR-factorization. To establish these
properties we recall the Gram-Schmidt algorithm for computing the
QR-factorization.
\begin{definition}[Gram-Schmidt]
\definitionlabel{GS}
The \emph{Gram-Schmidt orthonormalization} algorithm, denoted
$\mathrm{GS},$ is given
an input matrix $V\in\R^{d\times p}$ with columns $v_1,\dots,v_p$ and outputs
an orthonormal matrix $Q\in\R^{d\times p}$ with the same range as $V.$
The columns $q_1,\dots,q_p$ of $Q$ are computed as follows:

\noindent For $i=1$ to $p$ do:
\begin{itm}
\item $r_{ii} \leftarrow \|v_i\|$ 
\item $q_i \leftarrow v_i/r_{ii}$
\item For $j=i+1$ to $p$ do:
\begin{itm}
\item 
$r_{ij} \leftarrow \langle q_i,v_j\rangle$
\item $v_j \leftarrow v_j - r_{ij}q_i$
\end{itm}
\end{itm}
\end{definition}

The first states that the
Gram-Schmidt operation commutes with an orthonormal transformation of the
input.
\begin{lemma}
\lemmalabel{GSO}
Let $V\in\R^{d\times p}$ and let $O\in\R^{d\times d}$ be an orthonormal
matrix. Then, $\mathrm{GS}(OV) = O\times\mathrm{GS}(V).$
\end{lemma}
\begin{proof}
Let $\{r_{ij}\}_{ij\in[p]}$ 
denote the scalars computed by the Gram-Schmidt algorithm as specified in
\definitionref{GS}.
Notice that each of the numbers $\{r_{ij}\}_{ij\in[p]}$ is invariant under an
orthonormal transformation of the vectors $v_1,\dots,v_p.$ This is because
$\|Ov_i\|=\|v_i\|$ and $\langle Ov_i,Ov_j\rangle = \langle v_i,v_j\rangle.$
Moreover, The output $Q$ of Gram-Schmidt on input of $V$ satisfies $Q = VR,$
where $R$ is an upper right triangular matrix which only depends on the
numbers $\{r_{ij}\}_{i,j\in[p]}.$ Hence, the matrix $R$ is identical when the
input is $OV.$ Thus, $\mathrm{GS}(OV)=OVR=O\times\mathrm{GS}(V).$
\end{proof}

Given i.i.d. Gaussian matrices $G_0,G_1,\dots,G_L\sim N(0,1)^{d\times p},$
we can describe the behavior of our algorithm by a deterministic function
$f(G_0,G_1,\dots,G_L)$ which executes subspace iteration starting with $G_0$
and then suitably scales $G_\ell$ in each step. The next lemma shows that this
function is distributive with respect to orthonormal transformations.

\begin{lemma}
\lemmalabel{distrib}
Let $f\colon(\R^{d\times p})^L\to\R^{n\times p}$ denote the output of \PPM on
input of a matrix $A\in\R^{n\times n}$ as a function of the noise matrices
used by the algorithm as described above. Let $O$ be an orthonormal matrix
with the same eigenbasis as $A.$ Then,
\begin{equation}
f(OG_0,OG_1,\dots,OG_L)
= O\times f(G_0,\dots,G_L)\mper
\end{equation}
\end{lemma}

\begin{proof}
For ease of notation we will denote by $X_0,\dots,X_L$ the iterates of the
algorithm when the noise matrices are $G_0,\dots,G_L,$ and we denote by
$Y_0,\dots,Y_L$ the iterates of the algorithm when the noise matrices are
$OG_0,\dots,OG_L.$ In this notation, our goal is to show that $Y_L = O X_L.$

We will prove the claim by induction on~$L$.
For $L=0,$ the base case follows from \lemmaref{GSO}. Indeed,
\[
Y_0 = \mathrm{GS}(OG_0) = O \times \mathrm{GS}(G_0) = OX_0\mper
\]
Let $\ell\ge 1.$
We assume the claim holds for $\ell-1$ and show that it holds for
$\ell.$ We have,
\begin{align*}
Y_\ell & = \mathrm{GS}(AY_{\ell-1}+ OG_\ell) \\
&= \mathrm{GS}(AOX_{\ell-1} + OG_\ell) \tag{by induction hypothesis} \\
&= \mathrm{GS}(O(AX_{\ell-1} + G_\ell)) \tag{$A$ and $O$ commute} \\
&= O\times \mathrm{GS}(AX_{\ell-1} + G_\ell) \tag{\lemmaref{GSO}} \\
&= OX_\ell\mper
\end{align*}
Note that $A$ and $O$ commute, since they share the same eigenbasis by the
assumption of the lemma.  This is what we needed to prove. 
\end{proof}

The previous lemmas lead to the following result characterizing the
distribution of signs of inner products between the columns of $X_\ell$ and
the eigenvectors of~$A.$
\begin{lemma}[Sign Symmetry]
\lemmalabel{signs}
Let $A$ be a symmetric matrix given in its eigendecomposition as
$A=\sum_{i=1}^d\lambda_i u_i u_i^\trans.$
Let $\ell\ge 0$ and let $x$ be any column of $X_\ell,$ where $X_\ell$ is the
iterate of \PPM on input of $A.$
Put $S_i = \sign(\langle u_i,x\rangle)$ for $i\in[d].$
Then $(S_1,\dots,S_d)$ is uniformly distributed in $\signs^d.$
\end{lemma}

\begin{proof}
Let $(z_1,\dots,z_d)\in\signs^d$ be a uniformly random sign vector. Let 
$O=\sum_{i=1}^d z_i u_i u_i^\trans.$ Note that $O$ is an orthonormal
transformation. Clearly, any column $Ox$ of $OX_\ell$ satisfies the conclusion of
the lemma, since $\langle u_i,Ox\rangle = z_i \langle u_i,x\rangle.$
Since the Gaussian distribution is rotationally invariant, we
have that $OG_\ell$ and $G_\ell$ follow the same distribution. 
In particular, denoting by $Y_\ell$ the matrix computed by the algorithm if
$OG_0,\dots, OG_\ell$ were chosen, we have that $Y_\ell$ and $X_\ell$ are
identically distributed.
Finally, by \lemmaref{distrib}, we have that $Y_\ell = O X_\ell.$ By our
previous observation this means that $Y_\ell$ satisfies the conclusion of the
lemma. As $Y_\ell$ and $X_\ell$ are identically distributed, the claim also
holds for $X_\ell.$
\end{proof}

We will use the previous lemma to bound the $\ell_\infty$-norm of the
intermediate matrices $X_\ell$ arising in power iteration in terms of the
coherence of the input matrix. We need the following large deviation bound.

\begin{lemma}
\lemmalabel{sign-deviation}
Let $\alpha_1,\dots,\alpha_d$ be scalars such that $\sum_{i=1}^d\alpha_i^2=1$
and $u_1,\dots,u_d$ are unit vectors in $\R^n.$ Put
$B=\max_{i=1}^d\|u_i\|_\infty.$
Further let $(s_1,\dots,s_d)$ be chosen
uniformly at random in $\signs^d.$
Then,
\[
\Pr\Set{ \left\|\sum_{i=1}^ds_i\alpha_iu_i\right\|_\infty >
4B\sqrt{\log d}} \le 1/d^3\mper
\]
\end{lemma}

\begin{proof}
Let $X = \sum_{i=1}^d X_i$ where $X_i = s_i\alpha_iu_i.$ We will bound the
deviation of $X$ in each entry and then take a union bound over all entries.
Consider $Z=\sum_{i=1}^d Z_i$ where $Z_i$ is the first entry of $X_i$. The
argument is identical for all other entries of $X.$
We have $\E Z = 0$ and $\E Z^2 = \sum_{i=1}^d \E Z_i^2 \le
B^2\sum_{i=1}^d\alpha_i^2=B^2.$  Hence, by \theoremref{chernoff} (Chernoff
bound),
\[\textstyle
\Pr\Set{ \left|Z\right| > 4B\sqrt{\log(d)}} \le
\exp\left(-\frac{16B^2\log(d)}{4B^2}\right)
\le\exp(-4\log(d))=\frac1{d^4} \mper
\]
The claim follows by taking a union bound over all $d$ entries of $X.$
\end{proof}

\bibliographystyle{alpha}
\bibliography{moritz}

\appendix

\section{Deferred Concentration Inequalities}

\begin{theorem}[Chernoff bound]
\theoremlabel{chernoff}
Let the random variables $X_1,\dots,X_m$ be independent random variables such
that for every $i,$ $X_i\in[-1,1]$ almost surely. Let $X = \sum_{i=1}^m X_i$
and let $\sigma^2=\Var X.$ Then, for any $t>0,$ $\Pr\Set{ \left|X-\E X\right|
> t} \le \exp\left(-\frac{t^2}{4\sigma^2}\right)\mper$
\end{theorem}

The next lemma follows from standard concentration properties of the Gaussian
distribution.

\begin{lemma}\lemmalabel{gaussian-projection}
Let $U\in\mathbb{R}^{d\times k}$ be a matrix with orthonormal columns. 
Let $G_1,\dots,G_L\sim N(0,\sigma^2)^{d\times p}$ with $k\le p \le d$ and
assume that $L\le d.$ Then, with probability $1-10^{-4},$
\[
\max_{\ell\in[L]}\|U^\trans G_\ell\| \le O\left(\sigma \sqrt{p + \log L}\right)\mper
\]
\end{lemma}
\begin{proof}
  By rotational invariance of $G_\ell$ the spectral norm
  $\Norm{U^\trans G_\ell}$ is distributed like largest singular value
  of a random draw from $k\times p$ gaussian matrix
  $\mathrm{N}(0,\sigma^2)^{k\times p}.$ Since $p\ge k,$ the largest
  singular value strongly concentrates around $O(\sigma\sqrt{p})$ with
  a gaussian tail.  By the gaussian concentration of Lipschitz
  functions of gaussians, taking the maximum over $L$ gaussian
  matrices introduces an additive $O(\sigma \sqrt{\log L})$ term.
\end{proof}

We also have an analogue of the previous lemma for the Laplacian
distribution.

\begin{lemma}\lemmalabel{laplacian-projection}
Let $U\in\mathbb{R}^{n\times k}$ be a matrix with orthonormal columns. 
Let $G_1,\dots,G_L\sim \mathrm{Lap}(0,\lambda)^{d\times p}$ with $k\le p \le d$ and
assume that $L\le d.$ Then, with probability $1-10^{-4},$
\[
\max_{\ell\in[L]}\|U^\trans G_\ell\| \le O\left(\lambda \sqrt{pk}\log(Lpk)\right)\mper
\]
\end{lemma}

\begin{proof}
We claim that with probability $1-10^{-4}$ for every $\ell\in[L],$ every entry
of $U^\trans G_\ell$ 
is bounded by $O(\lambda\log(Lpk))$ in absolute value. This follows because each entry has
variance $\lambda^2$ and is a
weighted sum of $n$ independent Laplacian random variables
$\mathrm{Lap}(0,\lambda).$ Assuming this event occurs, we have
\[
\max_{\ell\in[L]}\|U^\trans G_\ell\| \le 
\max_{\ell\in[L]}\|U^\trans G_\ell\|_F \le 
O\left(\lambda \sqrt{pk}\log(Lpk)\right)\mper\qedhere
\]
\end{proof}

\begin{lemma}[Matrix Chernoff]\label{lem:matrixb2}
  Let $X_1, \dotsc, X_n \sim \cX$ be i.i.d. random matrices of
  maximum dimension $d$ and mean $\mu$, uniformly bounded by $\norm{X} \leq
  R$.  Then for all $t \leq 1$,
  \[\textstyle
  \Pr\Set{\Norm{\frac{1}{n}\sum_i X_i - \E X_1} \geq tR} \leq d e^{-\Omega(mt^2)}
  \]
\end{lemma}

\section{Reduction to symmetric matrices} 
For all our purposes it suffices to consider symmetric $n\times n$ matrices.
Given a non-symmetric $m\times n$ matrix $B$ we may always consider the
$(m+n)\times(m+n)$ matrix $A = [\,
0\, B\,
|\, B^\trans\, 0\,].$ This transformation preserves all the parameters that we are
interested in as was argued in~\cite{HardtR13} more formally. This allows us to discuss
symmetric eigendecompositions rather than singular vector decompositions and
therefore simplify our presentation below.

\end{document}